\documentclass[12pt]{article}
\usepackage{amsmath}
\usepackage{graphicx}
\RequirePackage[authoryear]{natbib}
\usepackage{url,color,subcaption,amsmath,amsthm} 
\theoremstyle{plain}
\newtheorem{theorem}{Theorem}[section]
\newtheorem{corollary}[theorem]{Corollary}

\newtheorem{lemma}[theorem]{Lemma}
\newcommand{\blind}{0}

\usepackage[linesnumbered,ruled]{algorithm2e}

\usepackage{multicol,booktabs, multirow,bm,tipa}
\begin{document}

\def\spacingset#1{\renewcommand{\baselinestretch}%
{#1}\small\normalsize} \spacingset{1}


\if0\blind
{
  \title{\bf A model-free subdata selection method for classification}
  \author{Rakhi Singh\hspace{.2cm}\\
    Department of Mathematics and Statistics, Binghamton University\\
    }
  \maketitle
} \fi

\if1\blind
{
  \bigskip
  \bigskip
  \bigskip
  \begin{center}
    {\LARGE\bf A model-free subdata selection method for classification}
\end{center}
  \medskip
} \fi

\bigskip
\begin{abstract}
Subdata selection is a study of methods that select a small representative sample of the big data, the analysis of which is fast and statistically efficient. The existing subdata selection methods assume that the big data can be reasonably modeled using an underlying model, such as a (multinomial) logistic regression for classification problems. These methods work extremely well when the underlying modeling assumption is correct but often yield poor results otherwise. In this paper, we propose a model-free subdata selection method for classification problems, and the resulting subdata is called PED subdata. The PED subdata uses decision trees to find a partition of the data, followed by selecting an appropriate sample from each component of the partition. Random forests are used for analyzing the selected subdata. Our method can be employed for a general number of classes in the response and for both categorical and continuous predictors. We show analytically that the PED subdata results in a smaller Gini than a uniform subdata. Further, we demonstrate that the PED subdata has higher classification accuracy than other competing methods through extensive simulated and real datasets.
\end{abstract}

\noindent%
{\it Keywords:} Big Data; Binary response; CART trees; Multi-class classification; Random forests; Subsampling; 
\vfill

\newpage
\spacingset{1.75} 

\section{Introduction}
\label{sec:intro}
The presence of substantially large datasets opens avenues for scientific inquiry and empirical progress in processes; concurrently, it introduces a multitude of challenges, for example, storage and analysis time. Presumably, despite the hardware advances, regular and everyday computing systems are inadequately equipped to expeditiously analyze such large datasets within reasonable time frames. Furthermore, even if such an analysis were plausible in theory, the feasibility is hindered by the impractical electricity consumption associated with sustained machine operations over a prolonged duration. Some statistical methods for analyzing big data include bags of little bootstraps by \cite{kleiner2014scalable}, divide-and-conquer \citep[][for example]{lin2011aggregated,chen2014split,song2015split,nicole2022} and sequential updating for streaming data \citep[][for example]{schifano2016online,xue2020online}. These methods are comprehensively reviewed in \cite{wang2016statistical}. Subdata selection is another novel area of research, which enables modeling and inference on a ``representative subset" of the data called \textit{subdata}. The field primarily focuses on choosing good subdata and analyzing it in a way that does not compromise statistical efficiency compared to analyzing the entire data.

Subdata selection methods often focus on identifying parameters of the underlying model, be it linear model \citep{ma2015leveraging, ma2015statistical, ting2018optimal, wang2019information, deng2023,singh_stufken_2023}, a logistic regression model \citep{wang2018optimal,cheng2020information}, non-linear models \citep{yu2023information}, softmax regression \citep{yao2023optimal}, hierarchical data using linear mixed effects model \citep{zhuwang2024}, or a Gaussian Process model \citep{chang2023predictive}. These approaches indirectly also identify a good subdata for prediction. For regression problems, \cite{joseph2021supervised} developed a data compression method that provides a systematic model-free alternative for prediction.  Other methods \citep{mak2018support, joseph2022split, vakayil2022data, zhang2023optimal} emulate the joint distribution of the big data. They are typically useful for splitting the dataset into training and test datasets where maintaining the same distribution among the different datasets is crucial. But, mimicking the same distribution as that of the big data need not be helpful for prediction or classification. For example, if the big data is heavy-tailed or has outliers, one can potentially build a better model by simply ignoring few extreme values. Other types of subdata selection methods \citep{meng2020more,shi2021model} are based on space-filling designs that intuitively aim to fill the $p$-dimensional space as possible, especially in terms of lower-dimensional projections. More recently, \cite{chang2024Predictive} developed a method for large-scale computer models based on expected improvement optimization, which is a model-free alternative to achieving superior prediction accuracy. The readers are referred to \cite{yu2023review} for a comprehensive review and a detailed comparison of these methods. 

In this article, our interest is in finding a good subdata for classification problems. Therefore, our options for existing subdata selection methods are limited. A few possibilities include probabilistic methods based on A-optimal or L-optimal designs for logistic and softmax regression \citep{wang2018optimal, yao2023optimal}, deterministic method based on the D-optimality of logistic regression \citep{cheng2020information}, and the support points based method to emulate the joint distribution of features \citep{vakayil2022data}. However, the first three methods are model-dependent and tend to perform poorly for a misspecified model. In addition, the underlying model is often unknown and more complex than a softmax regression. Thus, the subdata associated with a specific model is unlikely to perform well on real or simulated data. The methods that emulate the joint distribution of the feature space are more space-filling than the uniform distribution but may not lead to the best classification results, especially if there are sparser regions in the feature space. In addition, these methods typically only work for continuous features. The big data should probably be analyzed using more complex methods like decision trees or neural nets, especially if one cares more about the classification accuracy than the explainability of the models. If the data is to be analyzed using decision trees or neural nets, a model-free subdata selection method is necessary. However, to the best of the authors' knowledge,  model-free subdata selection methods do not exist for classification problems, especially for situations when some predictors are categorical.

This paper proposes a novel subdata selection method called Partition-Enabled Design of subdata (or PED, for short). Rooted in Proto-Indo-Aryan, PED (pronounced as \textipa{/pe:\:r/}) is a synonym for trees. The PED subdata first partitions the data using CART trees. Then, we select more observations from the component with the higher misclassification error in this initial partition. With PED subdata, we can decide on an appropriate over- and under-sampling proportion for each class. Additionally, our method provides a natural solution to handle a large number of predictor variables, with their nature being either continuous or categorical. While the other classification methods, such as neural nets, can be used on the PED subdata, our data is naturally suited more for datasets that work well with the classification and regression trees (CART), such as the datasets where the different classes can be identified from the axis-aligned partition in the $\mathbf{X}$-space. 

The rest of the paper is organized as follows. A motivating example is provided in Section~\ref{sec-motiv}. We then introduce the proposed method and study its properties and benefits in Section \ref{sec:method}. The performance of PED subdata is compared to other competing methods through simulations and real data in Section \ref{sec:numcom}. Finally, we end with the concluding remarks in Section \ref{sec:conc_remarks}.

\section{Motivating example}\label{sec-motiv}
We first provide an illustrative example of the target being simulated from a straightforward radial function for the two standard normal predictor variables. We consider 2500 observations in the full data. The left panel in Figure \ref{fig-mot} visually represents the response behavior. The majority class (92\%) is shown in the green circles, which are observations in the interior of the circle. The brown triangles (5\%) represent the observations outside the 95\% quantile of the radius of the circle.
In contrast, the remaining 3\% observations corresponding to the yellow squares are the ones that have a radius smaller than the 3\% quantile. The middle panel shows the 500 selected observations from a uniform random selection; the approximate percentages of the respective classes are the same as the full data. But, we don't select nearly enough observations from the two minority classes. The right panel shows the 500 selected subdata observations from the proposed PED method. We oversample the first minority class, undersample the second minority class, and undersample the majority class, resulting in the corresponding brown triangles (16\%), yellow squares (2.2\%), and green circles (81.8\%). 

\begin{figure}
\centering
\caption{The full data with three classes represented in different colors: green circles (92\%), brown triangles (5\%), and yellow squares (3\%). Middle panel has selected $n=500$ uniformly selected observations. The approximate class percentages for the middle panel are the same as that of the full data. The right panel has the $n=500$ selected observations from the PED subdata and the approximate class percentages are green circles (81.8\%), brown triangles (16\%), and yellow squares (2.2\%).}
\label{fig-mot}
\includegraphics[width=0.99\linewidth, trim=0cm 4cm 0cm 6cm,clip=true]{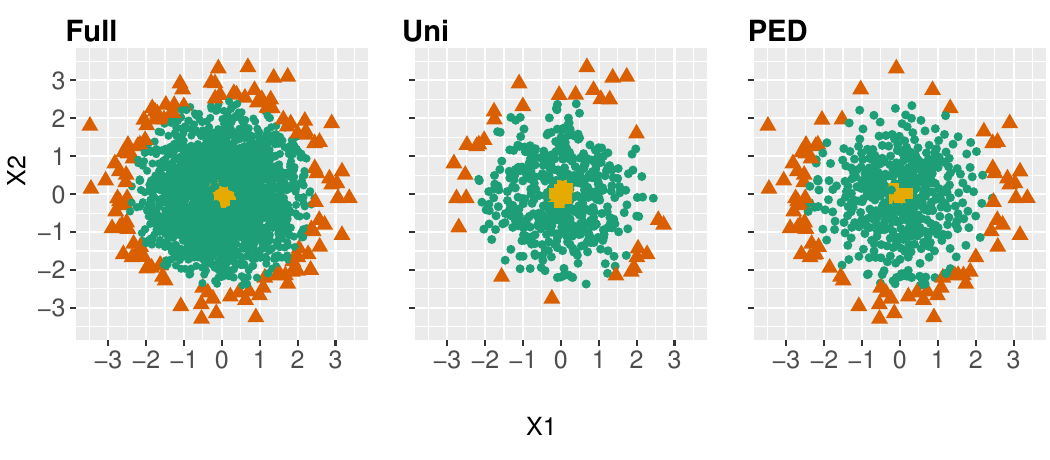}
\end{figure}

The subdata in the middle panel of Figure~\ref{fig-mot} is more likely to classify test observation in (say), $X1 \le -1$ and $X2 \in [-2,-1]$ as ``green circle", using random forests, which is not their actual class. Therefore, intuitively, methods like a uniform random selection, which emphasize a proportional representation of each category in the subdata, do not seem suitable for classification problems, especially when the class sizes are imbalanced. Our method is designed to ensure an appropriate representation of each class in the subdata, where the appropriateness is defined in terms of classification accuracy. Such optimization does not necessarily result in oversampling each minority class because the decision of oversampling classes also depends on the ease of classification. As in Figure~\ref{fig-mot}, PED subdata undersamples the ``yellow squares" because they are all clustered together in one center rectangle, and an axis-aligned partition method, such as CART trees, will easily classify these observations. In contrast, PED subdata exhibits a pronounced inclination to choose more observations along the left and right boundaries of $X1$ and $X2$. This preference arises due to the inherent difficulty CART trees encounter in discriminating between the ``green circles" and ``brown triangles". We will see in future sections that PED subdata, which, unlike uniform sample, does not retain the same proportion of classes in the subdata, leads to a smaller misclassification error or higher classification accuracy.

\section{Partition-enabled design of subdata}
\label{sec:method} We first provide a brief background on the random forests for classification, including the notations used throughout the paper. We then introduce our method and provide the supporting proofs, followed by the advantages and disadvantages of PEDs. 

\subsection{Brief background and notations}\label{sec-back}
TThe random forests construct a non-parametric, data-adaptive estimate of the unknown underlying function $g(.)$ using features. Since we are interested in a classification problem, the target outcome can take one of the $K$ values, $k=1,\dots,K$. Let $\mathcal{D} = (\mathbf{X},\mathbf{y})$ be the full data, where $\mathbf{X}$ is an $N\times p$ matrix with $N$ observations and $p$ independent variables or features and $\mathbf{y}$ is the corresponding $N\times 1$ target vector assuming one of the $K$ values. In what follows, we will use the terms features and response to denote $\mathbf{X}$ and $\mathbf{y}$, respectively. We also assume that the observations are independent. 

A CART tree finds a partition of the dataset $\mathcal{D}$ into $L$ strata, say, $\wp = (\wp_1,\dots, \wp_{\ell},\dots,\wp_{L})$ such that the 
\begin{equation}
\label{eq-minTree}
\sum_{\ell = 1}^{L}\sum_{i \in \wp_{\ell}} Error_{y_{i}, \hat{y}_{i}},
\end{equation}
is minimized, where $\hat{y}_{i}$ is the predicted class for the observations in the $\ell$th stratum and $i \in \wp_{\ell}$. The $\ell$th stratum corresponds to the $\ell$th leaf node of the tree, and both terms will be used interchangeably throughout this article. The ``Error" in \eqref{eq-minTree} can take several forms for classification, including the misclassification error, Gini index, and cross-entropy or deviance \citep[][see Section 9.2.3]{hastie2009elements}. The three metrics behave similarly, but the latter two are more commonly used since they are differentiable and more sensitive to the changes in the node probabilities. For $\ell$th node with $N_{\ell}$ observations, let $\hat{p}_{\ell, k}$ denotes the proportion of observations in the class $k$. Then, the Gini index for node $\ell$ is defined as 
\begin{equation}
\label{eq-Gini}
\sum_{k = 1}^{K}\hat{p}_{\ell, k} (1- \hat{p}_{\ell, k}).
\end{equation}
We use Gini throughout the paper. Since it is computationally infeasible to consider every possible partition \citep{james2013introduction}, trees are constructed using a greedy approach of recursive binary splitting. A recursive binary procedure splits a dataset into two subsets using a splitting criterion. Each split is characterized by the predictor variable and the splitting value on which a split is done. The predictor $X_j$ is randomly selected at each node from $p$ predictor variables. Let a split be made at a node $t$ on the splitting value $s$ of $j$th variable $X_j$. Then, the left and right daughter nodes of $t$ are $t_L$ and $t_R$ depending on whether $X_j \le s$ or $X_j > s$. We use the original Gini-based CART split criterion \citep{breiman1996bagging}. The splitting procedure is repeated until the tree height reaches a predetermined level or the terminal node has, at most, a pre-specified number of observations. The last grown nodes are called the terminal nodes or leaves of the tree. As noted earlier, a tree partitions the data; that is, a tree with $L$ (say) terminal nodes or leaves results in $L$ mutually exclusive and exhaustive pieces of $\mathcal{D}$. Every observation that falls into the terminal node $\ell$ will be assigned the predicted class equal to the mode of the response values in the node $\ell$. 

A random forest builds $ntree$ such trees, each time using a random (bootstrap) sample from $\mathcal{D}$. In addition, $X_j$ is chosen randomly from a randomly chosen set of $mtry$ predictors at each split. Then, a random forest classification at a test observation $\mathbf{x}_{test}$ is made by taking the majority vote from these $ntree$ trees. The $ntree$ is fundamentally a non-tunable parameter (increasing the number of trees always results in better classification accuracy), and the $mtry$ and other tuning parameters can be tuned for minor improvements in accuracy \citep{probst2019hyperparameters}. Typically, the default choices of hyperparameters in packages work well. We use the R package \texttt{ranger} \citep{wright2020ranger} and set $ntree$ and $mtry$ to be at 100 and $p/3$, respectively, along with other default choices of the hyperparameters.

\subsection{Description of PED methodology}
PED methodology first creates a CART-based coarse partition of the full data. Utilizing a stratified sampling approach, the number of samples to be selected from each stratum are then calculated such that the average Gini impurity for prediction on test data is minimized. Finally, twinning is employed within each stratum to select the desired number of observations. 

\subsubsection{Tree-based partition}
As mentioned in Section~\ref{sec-back}, CART trees partition the $X$-space by utilizing a loss function in the target space. Terminal nodes of fully grown trees tend to have a small Gini error, but they can lead to overfitting the training data. Therefore, the trees are typically grown to a certain depth, or the fully-grown trees are pruned. The trees with restricted depth are likely to have (a) a few terminal nodes with small Gini since no more splitting can be done in these nodes, and (b) the remaining nodes with large Gini since these nodes would have been split further if a deeper tree was allowed. As a result, trees of a certain depth tend to have different nodes with different Gini values. For the first step of the PED methodology, we strive to build a coarse partition of the full data. It is assumed that the available computational resources prohibit building CART trees on the full data all at once. It is also understood that a coarse partition is necessary at this stage to avoid overfitting, even though it yields a few strata and has a relatively larger Gini error than a potentially deeper tree. 

Using random subsets of size $t_s$ and all $p$ predictor variables, we build $t_n$ CART trees with maximum depths in $\{3,\dots, t_d\}$. Then, we use these $t_n(t_d-2)$ trees to identify the predicted class for the full data. These predicted values are then used to estimate the total Gini error for each of the $t_n(t_d-2)$ trees on the full data. The tree with the smallest total Gini error provides a partition of the full data and will be passed to the next stage. Note that in the next stage, this tree will be utilized to sample $n$ observations with at least $t_h$ observations in each stratum. Allowing any depth $t_d$ could result in a higher number of terminal nodes, making it impractical to sample from them all. Therefore, we restrict ourselves to finding the tree with the smallest Gini among all trees with no more than $n/t_h$ number of terminal nodes. The algorithm for the tree-based partition is formalized in Algorithm \ref{algo1}.

\begin{algorithm}[hbtp]
	\caption{Stage (I): Tree-based partition}\label{algo1}
	\SetKwInOut{Input}{inputs}
	\SetKwInOut{Output}{output}
	\Input{the sample size $n$, the size of a random sample $t_s$, the maximum depth for trees $t_d$, the number of trees for each depth $t_n$, the minimum number of observations to be considered in sample from each stratum $t_h$, the full data $(\mathbf X, \mathbf y)$ and its size $N$} 
	\For{$ j_1=1 \rightarrow t_n $}{
   	   	\For{$ j_2=3 \rightarrow t_d $}{
   	   Let $\mathbf{X}_{t_s,p}$ be a $t_s\times p$ matrix based on a uniform random sample of size $t_s$ from the rows of $\mathbf{X}$,  and $\mathbf{y}_{t_s}$ be the corresponding target vector\;
   	   Build a CART tree on $\mathbf{y}_{t_s}$ and $\mathbf{X}_{t_s,p}$, with tuning parameters $mtry = p$, $max.depth = j_2$ and other default parameters\;
   	   Use the tree built in line 4 to identify the node for each observation in the full data $(\mathbf X, \mathbf y)$\; 			   Compute the number of items ($N_{\ell}$, say), mean, and Gini ($G_{\ell}$, say) in each node\;
   	   Compute the total Gini of the full data based on the predicted classes obtained by fitting a tree to full data in line 5\;
       Save the Gini, node labels for each observation, and $(N_{\ell}, G_{\ell})$ for each tree\;
}
}
	Find the index of tree with the smallest total Gini among all trees with the number of terminal nodes no more than $n/t_h$\;
	\Output{Return the best tree with its node labels and $(N_{\ell}, G_{\ell})$ } 
\end{algorithm}

Line 2 of Algorithm \ref{algo1} shows that the depth of our trees is restricted. With a restricted depth, we get a coarse partition that avoids overfitting. For example, Figure \ref{fig-part} illustrates the effect of varying the maximum depth for the response function described in Section~\ref{sec-motiv}. Each color represents a different stratum, and we can observe that as the depth of the tree increases, the identification of different classes improves. Increasing the depth beyond a point leads to excessive splitting in the areas corresponding to the majority class in Figure~\ref{fig-mot}, which may lead to overfitting and is less desirable. In line 4 of Algorithm \ref{algo1}, the trees are constructed on a dataset of size $t_s$, where $t_s \ll N$. Lines 5--7 of Algorithm \ref{algo1} predict the classes of the target on the full data using the tree built on a small sample. While one cannot build the tree on the full data, classifying a test example is usually cheaper. This generalization to the full data enables us to identify the tree structure that works best for the full data. Since the shallow trees are susceptible to the data, the term ``best" tree is loosely defined. We will see in Section \ref{sec:numcom} that such a partition results in a good classification performance. 

\begin{figure}[hbtp]
\centering
\caption{Strata using different maximum depths for the example in Figure~\ref{fig-mot}. Different colors represent observations in different strata. Note that the partition is axis-aligned and the edges in the plot look rough because of plotting restrictions.}
\label{fig-part}
\includegraphics[width=0.99\linewidth, trim=0cm 4cm 0cm 6cm,clip=true]{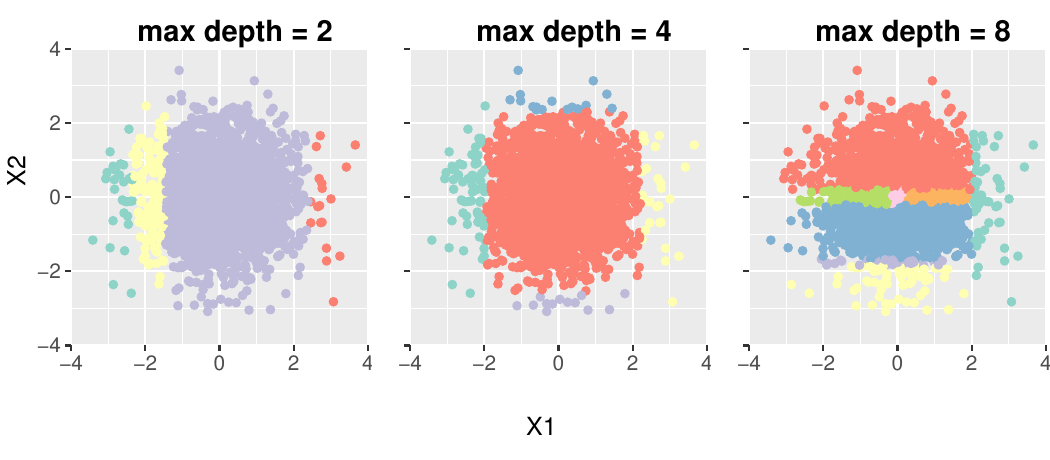}
\end{figure}
\vspace{-1cm} 

\subsubsection{Stratified sampling utilizing the partition}
Let the full data $\mathcal{D} = (\mathbf X, \mathbf y)$ be partitioned in $L$ strata using Algorithm \ref{algo1}. Let the partition be denoted by $\wp = (\wp_1,\dots, \wp_{\ell}, \dots, \wp_L)$. Note that Algorithm~\ref{algo1} returns the partition $\wp$ along with the number of observations $N_{\ell}$ and Gini $G_{\ell}$ for the stratum $\ell$ in the full data, $\ell = 1,\dots,L$. We use a stratified sampling approach to determine the number of samples to be selected from each stratum. Let $n_{\ell}$ denote the number of observations to be sampled from $\ell$th stratum of size $N_{\ell}$.  Since the test data is expected to utilize the same data generation mechanism as the training data, it is reasonable to assume that the partition obtained in Algorithm \ref{algo1} on the training data will generalize well to the test data. With this assumption, Theorem \ref{th1} finds an optimal sample size per stratum. 

\begin{theorem}
\label{th1}
Let the full dataset can be partitioned into $\wp = (\wp_1,\dots, \wp_{\ell}, \dots, \wp_L)$ where the $\ell$th stratum has $N_{\ell}$ observations with the Gini $G_{\ell}$. Assume that the test dataset can be characterized through the same partition structure as provided by $\wp$. Then, the total expected Gini on the test dataset will be minimized when $$n_{\ell} \propto \sqrt{N_{\ell}\;G_{\ell}},$$
where $n_{\ell}$ observations are included in subdata from the $\ell$th stratum and $\sum n_{\ell} = n$.
\end{theorem}
\begin{proof}
With the assumption that the test dataset can be characterized through the same partition as the full data, the total expected Gini on a test dataset is
\begin{equation}
  \mathcal{G} = E\bigg[\sum_{\ell = 1}^L \bigg(\frac{N_{\ell}}{N}\bigg) \sum_{k = 1}^{K}\hat{p}_{\ell, k} (1- \hat{p}_{\ell, k}\big) \bigg], \label{eq-predgini}
\end{equation} 
where $\hat{p}_{\ell, k}$ denotes the predicted proportion of observations in the class $k$ and stratum $\ell$, where predictions are obtained by using the subdata. The weights $\frac{N_{\ell}}{N}$ in \eqref{eq-predgini} ensure that the distribution of the number of observations in the partition $\wp$ remains the same between the full and the test data. Now, $E[\hat{p}_{\ell, k}] = p_{\ell, k}$ and $V[\hat{p}_{\ell, k}] = \frac{p_{\ell, k}(1 - p_{\ell, k})}{n_{\ell}}$, where  $p_{\ell, k}$ is the proportion of observations in the class $k$ and stratum $\ell$ from the full data. Therefore, $E[\hat{p}^2_{\ell, k}]$ simplifies to $\frac{p_{\ell, k}(1 - p_{\ell, k})}{n_{\ell}} + p_{\ell, k}^2$. Then, the total expected Gini in \eqref{eq-predgini} becomes
\begin{equation}
  \mathcal{G} = \sum_{\ell = 1}^L \bigg\{1 - \sum_{k = 1}^{K}{p}^2_{\ell, k} - \sum_{k = 1}^{K}\frac{p_{\ell, k}(1 - p_{\ell, k})}{n_{\ell}}\bigg\} . \label{eq-predvaradapted}
\end{equation} 
Using the Lagrange multiplier to minimize $\mathcal{G}$ with the constraint that $\sum_{\ell} n_{\ell} = n$, we get the desired result.  
\end{proof}

\begin{corollary}\label{cor1}
Subdata with sample size $n_{\ell} \propto \sqrt{N_{\ell}\;G_{\ell}}$  has a smaller total Gini as defined in  \eqref{eq-predgini} than the subdata with $n_{\ell} \propto N_{\ell}$, which corresponds to the uniform random selection.
\end{corollary}

Two problems may arise in practice: (1) $G_{\ell}$ could be 0, and (2) few $N_{\ell} < n_{\ell}$. Even though (1) is less likely, it is theoretically possible to achieve a pure node (with $G_{\ell}=0$) or a node with $N_{\ell}=1$. Therefore, in practice, we find the optimal values of  $n_{\ell} \propto \sqrt{N_{\ell}\;G_{\ell}}$ such that $\min\{t_h, N_{\ell}\} \le n_{\ell} \le N_{\ell}$. 

Selecting very few observations randomly from such stratum yields a relatively unstable subdata since a uniform sample does not satisfactorily mimic the population distribution for small sample sizes \citep{mak2018support}. twinning \citep{vakayil2022data} is a fast sampling method that ensures that the subdata has the same joint distribution as the full data. Utilizing twinning is a worthwhile cause, especially since some of our $n_{\ell}$ can be very small. Once the $n_{\ell}$ values are calculated based on Theorem~\ref{th1}, we utilize twinning within each stratum to collect the observations in the subdata. Therefore, for the $i$th stratum, we select $n_{\ell}$ observations from $N_{\ell}$ using twinning. Algorithm \ref{algo3} outlines PED methodology starting from the partition obtained via Algorithm \ref{algo1}. The \texttt{twinning} package in R may not give exactly $n_{\ell}$ observations depending on the value of the ratio $r_{\ell} = N_{\ell}/n_{\ell}$. If $r_{\ell}$ is integer, then twinning selects exactly $n_{\ell}$ observations. If $r_{\ell}$ is rounded up to the nearest integer, one selects less than $n_{\ell}$ observations. However, if $r_{\ell}$ is rounded down to the nearest integer, one selects more than $n_{\ell}$ observations. In Algorithm \ref{algo3}, we round down $r_{\ell}$ to its nearest integer, and then we randomly select $n_{\ell}$ observations. Therefore, the PED subdata in each stratum has the same joint distribution as the corresponding population stratum. Corollary \ref{cor1} implies that the PED subdata obtained from the optimal sample sizes as per Theorem \ref{th1} would perform better than the twinning or the uniform subdata. We demonstrate this via numerical comparisons in Section \ref{sec:numcom}.

\begin{algorithm}[hbtp]
	\caption{Partition-Enabled Design (PED) of subdata }\label{algo3}
	\SetKwInOut{Input}{inputs}
	\SetKwInOut{Output}{output}
	\Input{the sample size $n$, the size of a random sample $t_s$, the maximum depth for trees $t_d$, the minimum number of observations to be considered in sample from each stratum $t_h$, the number of trees for each depth $t_n$, the full dataset $(\mathbf X, \mathbf y)$ and its size $N$} 
	Use Algorithm \ref{algo1} to find the partition of the full data\;
	Use the partition in Step 1 to find the sample sizes, $n_{\ell}$ (say) such that $n_{\ell}$ is as in Theorem~\ref{th1}\;
    Let $r_{\ell} = \lfloor N_{\ell}/n_{\ell} \rfloor$
	\For{all strata}{
            Select $n'_{\ell}$ observations from each node using \textsc{twinning} using $r_{\ell}$\;
            \If{$n'_{\ell} >= n_{\ell}$}{
             Randomly select $n_{\ell}$ observations from the $n'_{\ell}$ selected observations in step 4.
             }
	}		
	\Output{Return the subdata} 
\end{algorithm}

\subsection{Computational complexity of PED} 
\cite{sani2018computational} showed that the computational complexity of building a CART tree on $n$ observations with $m$ variables is $O(mn \log{n})$. Further, building a random forest on full data consisting of $t$ trees is $O(tpN\log{N})$. Theorem~\ref{thComp} provides a worst-case complexity of obtaining the PED subdata and building a random forest on the obtained subdata. But, we first need the following lemma. 

\begin{lemma}\label{lemComp}
The sum of the average time complexity of twinning each of the two component datasets of size $rN$ and $(1-r)N$ for $r \in [0,1]$  is at most than that of twinning the entire dataset. 
\end{lemma}
\begin{proof}
 The average time complexity for twinning on $n$ observations in $m$ variables is $O(mn\log{n})$. Suppose the dataset with $N$ observations is divided into two subsets of size $N_1 = rN$ and $N_2 = (1-r)N$, with  $r \in [0,1]$. From the convexity argument, it follows that 
\begin{eqnarray}
    \nonumber && Nr\log(Nr) + dN(1-r)\log(N(1-r)) \\
    \nonumber &= &dN\{rlogN + rlogr + \log(N) + \log(1-r) - r\log(N) - r\log(1-r)\}\\
   \nonumber  & = & dN\{\log(N) + r\log(r) + (1-r)\log(1-r)\}\\
    \nonumber & \le &  dN\log(N),
\end{eqnarray}
where the last inequality follows from the properties of the log function for $r \in [0,1]$.
\end{proof}

We are now ready to prove the following result. 

\begin{theorem}\label{thComp}
The random forest fitted on the PED subdata obtained in Algorithm~\ref{algo3} has the overall worst-case computational complexity of 

$$O(t(p\sqrt{N} \log(\sqrt{N}) + N \log(\sqrt{N}))+ pN \log{N}),$$  
where $t$ is the total number of trees built in Algorithm~\ref{algo1}.
\end{theorem}
\begin{proof}
For each tree (outer for loop in Step 1), Step 4 of Algorithm \ref{algo1} has the computational complexity $O(p\sqrt{N} \log{\sqrt{N} })$. We have trees with the depth at most $\log(\sqrt{N})$ nodes. The $\log(\sqrt{N})$ comparisons are made for a test observation to reach a terminal node. Therefore, obtaining predictions for $N$ observations from each tree will have a worst-case complexity of $O(N\log(\sqrt{N}))$. Steps 6--7 will then have at most complexity of $O(N)$ for each tree. Overall, Steps 1--9 of Algorithm \ref{algo1} has the complexity of 
$O(t(p\sqrt{N} \log(\sqrt{N})+ N\log(\sqrt{N}))$. Selecting the best tree has $O(1)$ complexity. Therefore, with $t = t_n (t_d-2)$ being the number of constructed trees, overall Algorithm \ref{algo1} has the complexity of $O(t(p\sqrt{N} \log(\sqrt{N}) + N \log(\sqrt{N})))$.

Selecting the number of correct observations is a relatively cheap operation and has the complexity of $O(\log(\sqrt{N}))$. From Lemma \ref{lemComp}, we see that the average complexity of twinning two component datasets is upper bounded by the average complexity of twinning the entire dataset. If Algorithm~\ref{algo1} partitions the dataset with $N$ observations in vaguely equal strata, then using parallelization, the twinning complexity can be reduced to $O(p\sqrt{N}\log{N})$. The parallelization does not help if the dataset is such that Algorithm~\ref{algo1} results in one major stratum and several other small strata. For such a case, the worst-case complexity for the twinning is $O(pN\log{N})$. Therefore, PED subdata has the overall worst-case computational complexity of $O(t(p\sqrt{N} \log(\sqrt{N}) + N \log(\sqrt{N}))+ pN \log{N}).$

Further, since building a random forest on PED subdata with $t$ trees is $O(tpn\log{n})$, the overall complexity is the sum of $O(t(p\sqrt{N} \log(\sqrt{N}) + N \log(\sqrt{N}))+ pN \log{N})$ and $O(tpn\log{n})$. Finally, since $O(tpn\log{n})$ is a dominant term, we get the desired result.
\end{proof} 

Assuming that $N\gg p \ge t/2$ and, in particular, that $t/2 \le p\le \sqrt{N}$, the overall computational complexity reduces to $O(pN \log{N})$, which is the same as that of the twinning. Note also that the computational complexity of running random forests on the full data is $O(tpN \log{N})$. Interestingly, the twinning, the random forests on full data, and the random forests built on the PED subdata all have the same complexity, assuming $t$ to be a constant. 

The actual computational time of PED subdata can be further reduced if (a) trees in lines 2-9 of Algorithm~\ref{algo1}, and (b) twinning on different subdata, can be run in parallel. In addition, note that the line 7 of Algorithm~\ref{algo1} demands to compute the total Gini on the full dataset, which is unnecessary and computationally challenging for huge datasets. We consider a random sample of size $5\times 10^6$ (= 5M) if the full dataset has more than 5M observations for our implementation. Further, one can also take uniform samples instead of twinning within each stratum, especially if $n_{\ell}$'s are close to the corresponding $N_{\ell}$ or both $n_{\ell}$ and $N_{\ell}$ are large. The random forest package in R, \texttt{ranger}, is fast and uses parallelization. Our implementation of PED subdata also utilizes parallel computing. In Section~\ref{sec:numcom}, via numerical comparisons, we will show that the PED subdata is faster than the competing alternatives, including twinning and the random forests on the full data, especially for large $N$ and $p$. All logs in this section were to the base 2. 

\subsection{Usage and benefits of PED}
The tuning parameters are an essential consideration in the practical use of any method. PED has four tuning parameters: the size of the random sample on which the partition is created, $t_s$; the maximum depth for trees that determine the right partition, $t_d$; the minimum number of observations to be considered in a sample from each stratum, $t_h$; and the number of trees built for each depth, $t_n$. The two most important parameters are $t_s$ and $t_d$. Large values of $t_s$ and $t_d$ result in a partition that overfits the initial random sample taken to create the partition and does not generalize well to the full data. Large values for $t_s$ and $t_d$ also result in a significant increase in the computation time. Consistency results for random forests and CART trees with additional additivity assumption on the regression function are studied in \cite{scornet2015consistency} and \cite{klusowski2021universal}. These analytical studies suggest that the subsampling ratio and the depth of the tree are important tuning parameters that affect the consistency of the predictions. The depth of the order of $\log_2(n)$ provides consistent prediction estimates, where trees are built on $n$ observations. We suggest using $t_s = \sqrt{N}$ and $t_d = \log_2(t_s)$ as the default values for the tuning parameters in Algorithm \ref{algo1}. Though rooted in regression trees, our parameter tuning methodology demonstrates its efficacy in addressing the classification problems. 

Algorithm \ref{algo1} also uses $t_n$ and $t_h$ as tuning parameters. Increasing the number of trees $t_n$ results in a larger computation time, and the corresponding statistical gain is minimal; therefore, fixing $t_n$ at a small value is desirable. We find that building ten trees of each depth, so $t_n= 10$ is both time-efficient and sufficient. The threshold parameter $t_h$ is consistent with the \texttt{maxnodes} parameter of random forests, which helps in preventing overfitting. The threshold $t_h$ ensures the selection of at least a few observations from a relatively pure stratum. Without setting a threshold, one may have zero observations from such strata, which is undesirable. We set $t_h = 5$ as is the case with random forests. While we do not provide theoretical justification for these settings, we see in Section \ref{sec:numcom} that these choices consistently deliver good results. 

PED is a model-free subdata selection method that utilizes the relationship between the predictor and response variables. It is computationally and statistically efficient, works for a large number of a mixture of categorical and continuous variables, and has no model-based assumptions. Having either of these characteristics is desirable, and each subdata selection method discussed here lacks one or the other elements. By being built on a CART tree-based methodology, PED satisfies all these characteristics and is a powerful method. PED has the same limitations as CART trees. The slight variations in the dataset can result in significant changes in the tree structure, causing our Algorithm \ref{algo1} to give a different partition. The PED subdata performs well for all our simulated studies despite this instability.

\section{Illustration}
\label{sec:numcom}
We evaluate PED subdata on several simulated and real datasets by fitting random forests and measure the corresponding output in terms of the accuracy of classification, area under curve (AUC), and computation times. 

\subsection{Simulated data}
We consider several examples, some common to the random forest literature and others to the subdata selection literature. We compare the performance of the PED subdata to that of twinning and the uniform subdata. We also compare the model-based logistic IBOSS method of \cite{cheng2020information} wherever possible. After collecting the subdata, we fit a random forest model with 100 trees and $mtry = p/3$ if $p\ge 5$ or $p$ otherwise. Here $mtry$ denotes the number of features used for each split in a tree.  We use the R package \texttt{ranger} with the default values for other hyperparameters. For all simulated studies in this section, we use the big data size $N = 10^5$, the validation dataset size equals $10^4$, and the two sample sizes considered are $1\%$ and $5\%$ of $N$, that is, 1000 and 5000. The simulation is repeated 50 times, implying that new training and test data are generated 50 times following the same data generation mechanism. In addition, the PED subdata uses the following values of the hyperparameters across all simulations: $t_s =\sqrt{N}$, $t_d = \log_2{t_s}$, $t_h=5$, and $t_n = 10$.

We begin by evaluating the performance of PED subdata on the motivating example discussed in Section~\ref{sec-motiv}. We have $p= 2$ independent features, which follow a standard normal distribution. The visual representation of the data is provided in the left panel of Figure~\ref{fig-mot}. For this example, the majority class (92\%) consists of the observations in the interior of the circle. The second class  (5\%) consists of the observations that lie outside the 95\% quantile of the radius of the circle, whereas the remaining 3\% observations are the ones that have a radius smaller than the 3\% quantile. Figure~\ref{fig-res} shows the performance of PED subdata compared to the full data, twinning, and uniform subdata on the validation dataset. The left panel shows the boxplot of the accuracy, whereas the right panel shows the boxplots of AUC, where the boxplots measure the variation across 50 iterations of the data generation mechanism. In both panels, the first set of boxplots uses either the full data (in pink) or the different subdata with $0.01N = 1000$ observations. In contrast, the right set demonstrates the performances for the subdata with $0.05N = 5000$ observations. The PED subdata performs superior to other subdata selection methods for both sample sizes. We also note that the performances for all subdata methods improve with larger sample sizes.  

\begin{figure}[hbtp]
\centering
\caption{Accuracy and AUC for the three-class example illustrated in Figure~\ref{fig-mot}. The x-axis for both plots demonstrates the subdata sizes for the three subdata selection methods: PED, Uni, and twinning.}
\label{fig-res}
\includegraphics[width=0.99\linewidth, trim=0cm 3cm 0cm 3cm,clip=true]{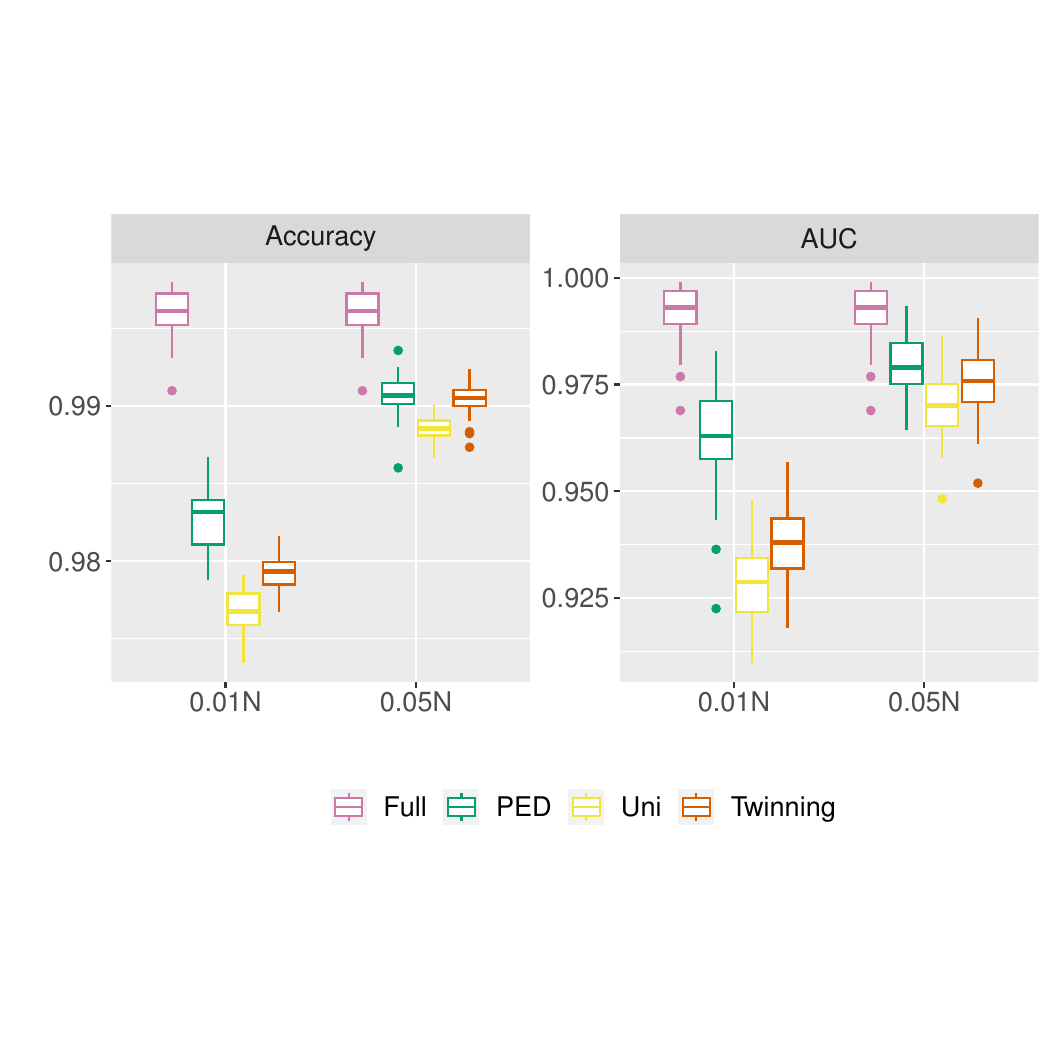}
\end{figure}

Next, we consider different classification datasets commonly used in the machine learning literature, especially with the tree-based classifiers \citep{breiman1996bias}. We consider the \texttt{waveform} function with $p=21$ features and $K=3$ classes and the \texttt{threenorm} function with $p=20$ features and $K=2$ classes. We also consider \texttt{threenorm}, \texttt{twonorm}, and \texttt{ringnorm} response functions with $p=2$ features and $K=2$ classes. We also consider the imbalanced version of the \texttt{threenorm} function, with 95\% of the data coming from one class and the remaining 5\% from another. The detailed definitions of these functions are provided in the Supplementary Material. We generate the data by using the R package \texttt{mlbench} \citep{mlbenchPackage}. Table~\ref{tabML1Acc} reports the classification accuracy for the different datasets and methods, whereas the corresponding AUC values are provided in the Supplementary Material. The higher the accuracy, the better the method. For most settings, random forests built on the full dataset have a larger accuracy than any subdata selection methods, which only use 1\% or 5\% of the available observations. While the accuracy numbers are fairly close, the PED subdata generally performs slightly better than the corresponding uniform or twinning dataset. 

\begin{table}[htbp]
  \centering
  \caption{Accuracy (in \%) for multiple simulated machine learning datasets with $p$ features and $K$ classes for the full dataset as well as the PED, uniform, and twinning subdata with sample sizes, $n = 0.1N$ and $0.05N$.}
    \begin{tabular}{|l|r|r|r|r|r|r|r|r|r|}\hline
       &  &  & & \multicolumn{3}{c|}{$n = 0.01N$} & \multicolumn{3}{c|}{$n = 0.05N$} \\
\cmidrule{5-10}  & $p$ & $K$&   \multicolumn{1}{c|}{Full}  & \multicolumn{1}{c|}{PED} & \multicolumn{1}{c|}{Uni} &  \multicolumn{1}{c|}{Twin}  & \multicolumn{1}{c|}{PED} & \multicolumn{1}{c|}{Uni} &\multicolumn{1}{c|}{Twin}  \\\hline
        Waveform  & 21    & 3     & 85.54 & 83.92 & 83.81 & 83.85 & 84.83 & 84.78 & 84.80 \\
    Threenorm & 2     & 2     & 88.05 & 88.05 & 87.35 & 87.31 & 88.26 & 87.71 & 87.69 \\
    Threenorm & 20    & 2     & 88.25 & 85.04 & 85.02 & 85.38 & 86.75 & 87.03 & 87.18 \\
    Imbalanced threenorm & 2     & 2     & 96.30 & 96.23 & 96.07 & 96.05 & 96.25 & 96.17 & 96.17 \\
    Imbalanced threenorm & 20    & 2     & 97.68 & 95.86 & 95.21 & 95.19 & 97.00 & 96.25 & 96.27 \\
    Ringnorm & 2     & 2     & 71.83 & 72.72 & 71.20 & 71.09 & 72.90 & 71.43 & 71.35 \\
    Twonorm & 2     & 2     & 97.39 & 97.30 & 97.10 & 97.13 & 97.40 & 97.29 & 97.28 \\
   \hline \end{tabular}%
  \label{tabML1Acc}%
\end{table}

Now, we consider the simulated datasets common in the subdata selection literature for categorical responses. These methods have an underlying model assumption, and the data is typically generated using either a logistic regression model or a softmax regression model, depending on the number of considered classes. For binary response, a logistic regression model is considered with seven features having the true $\bm{\beta} = (0.5,0.5,0.5,0.5,0.5,0.5,0.5)^T$. The variance-covariance matrix $\bm{\Sigma}$ is a $7\times 7$ matrix with diagonal elements one and off-diagonal elements equal to 0.5.  Following \cite{cheng2020information}, we consider the following covariate distributions $\mathcal{N}_7(0,\bm{\Sigma})$, $\mathcal{N}_7(1,\bm{\Sigma})$, $0.5\mathcal{N}_7(1,\bm{\Sigma})+0.5\mathcal{N}_7(-1,\bm{\Sigma})$, and $\bm{t}_3(0,\bm{\Sigma}/10)$. Here $\mathcal{N}_3$ denote a multivariate normal distribution, whereas $\bm{t}_3$ denote a multivariate $ t$ distribution with 3 degrees of freedom. The four cases are called Bin-MVN0, Bin-MVN1, Bin-Mix, and Bin-T3, respectively. The choice of different feature distributions leads to an imbalance in class distributions, further details of which are provided in the Supplementary Material. Table~\ref{tabBinary1Acc} reports the classification accuracy for these four covariate distributions, and the corresponding AUC values are provided in the Supplementary Material. In addition to the PED and uniform subdata, we add a random forest fitted on the IBOSS  \citep{cheng2020information} subdata to the comparison. The performance of the twinning subdata is omitted due to its extremely similar performance to the uniform subdata. We again observe that the PED subdata has larger accuracy than other subdata selection methods and is close to the accuracy of random forests built on the full data. 

\begin{table}[htbp]
  \centering
  \caption{Accuracy (in \%) for multiple simulated logistic regression datasets with $7$ features and $2$ classes for the full dataset as well as the PED, uniform, and IBOSS subdata with sample sizes, $n = 0.1N$ and $0.05N$.}
    \begin{tabular}{|l|r|r|r|r|r|r|r|}\hline
       &   & \multicolumn{3}{c|}{$n = 0.01N$} & \multicolumn{3}{c|}{$n = 0.05N$} \\
\cmidrule{2-8} &   \multicolumn{1}{c|}{Full}  & \multicolumn{1}{c|}{PED} & \multicolumn{1}{c|}{Uni} &  \multicolumn{1}{c|}{IBOSS}  & \multicolumn{1}{c|}{PED} & \multicolumn{1}{c|}{Uni} &\multicolumn{1}{c|}{IBOSS}  \\\hline
    Bin-MVN0 & 81.23 & 80.78 & 80.54 & 80.57 & 81.15 & 80.94 & 80.94 \\
    Bin-MVN1 &90.43 & 90.25 & 89.98 & 87.19 & 90.38 & 90.24 & 89.45 \\
    Bin-Mix & 76.07 & 75.60 & 75.26 & 75.29 & 75.97 & 75.71 & 75.73 \\
    Bin-T3 & 66.78 & 66.10 & 65.55 & 65.56 & 66.58 & 66.17 & 66.22 \\\hline 
    \end{tabular}%
  \label{tabBinary1Acc}%
\end{table}

Finally, we consider datasets with three classes. Following \cite{yao2023optimal}, we consider four cases for the covariate distribution. The data is generated from a softmax regression with three features, and the true parameter vector is $\bm{\beta} = (1,1,1,2,2,2)^T$. The variance-covariance matrix $\bm{\Sigma}$ is a $3\times 3$ matrix with diagonal elements one and off-diagonal elements equal to 0.5.  The four different cases are such that the predictors follow $\mathcal{N}_3(0,\bm{\Sigma})$, $\mathcal{N}_3(1.5,\bm{\Sigma})$, $0.5\mathcal{N}_3(1,\bm{\Sigma})+0.5\mathcal{N}_3(-1,\bm{\Sigma})$, and $\bm{t}_3(0,\bm{\Sigma})$ distributions, respectively. The four cases are called Mult-MVN0, Mult-MVN1.5, Mult-Mix, and Mult-T3, respectively. More details and the AUC results are again provided in the Supplementary Material. Figure~\ref{fig-resMult} demonstrates that the PED subdata performs superior to both uniform and twinning subdata. The PED subdata also performs close to that of the full data.

\begin{figure}[hbtp]
\centering
\caption{Accuracy for the simulated datasets from the softmax regression. The x-axis for all subplots demonstrates the subdata sizes for the three subdata selection methods: PED, Uni, and twinning.}
\label{fig-resMult}
\includegraphics[width=0.8\linewidth, trim=0cm 0cm 0cm 0cm,clip=true]{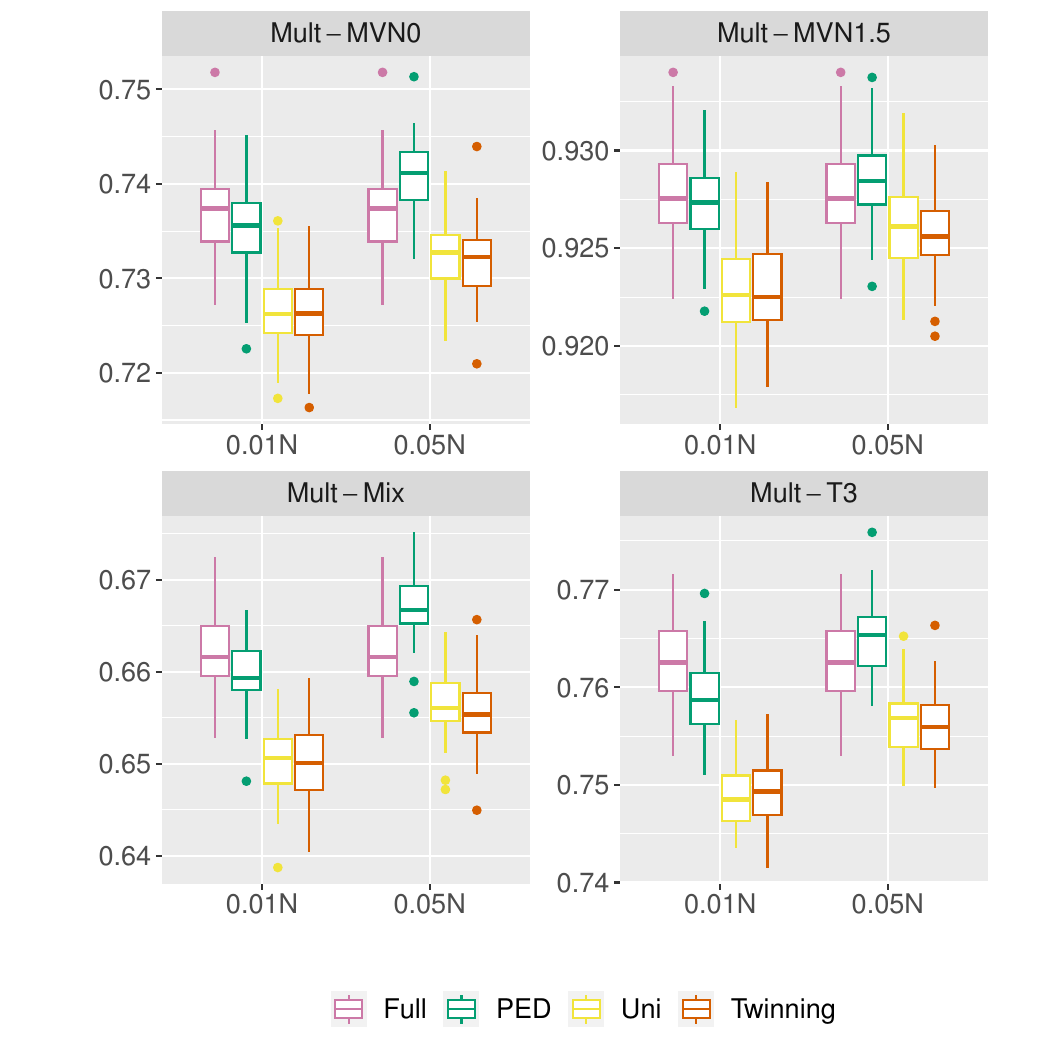}
\end{figure}

As for the computation times, we vary the size of the training dataset from $10^5$ to $10^7$  and the number of features in Tables~\ref{tabTimevaryingN} and ~\ref{tabTimevaryingK}, respectively. Both tables demonstrate that the twinning subdata has the largest computation time. The uniform subdata takes the smallest runtime, but the performance can often be increased using the PED subdata. The PED subdata also has considerably smaller runtimes than the random forests built on the full data, which is advantageous for the bigger datasets. The computation times for the PED subdata are not dependent on the size of the subdata compared to the twinning subdata. All computations are done on a Desktop with an AMD Ryzen Threadripper PRO 5955WX @4.00 GHz 16 cores and 64GB RAM.

\begin{table}[htbp]
  \centering
  \caption{Computation time (in mins) for the full data and PED, uniform, and twinning subdata with sample sizes, $n = 0.1N$ and $0.05N$. The number of features are fixed at $p=100$ and the number of observations vary from $10^5$ to $10^7$. The $^\dag$ represents that the program was exited after the said time.}
    \begin{tabular}{|l|r|r|r|r|r|r|r|}\hline
       &   & \multicolumn{3}{c|}{$n = 0.01N$} & \multicolumn{3}{c|}{$n = 0.05N$} \\
\cmidrule{2-8} &   \multicolumn{1}{c|}{Full}  & \multicolumn{1}{c|}{PED} & \multicolumn{1}{c|}{Uni} &  \multicolumn{1}{c|}{Twin}  & \multicolumn{1}{c|}{PED} & \multicolumn{1}{c|}{Uni} &\multicolumn{1}{c|}{Twin}  \\\hline
      $N =10^5$ & 0.22  & 0.32  & 0.01  & 0.17  & 0.33  & 0.01  &0.74 \\
   $N =10^6$ & 10.99 & 3.33  & 0.02  & 33.81 & 4.69  & 0.13  & 166.52 \\
    $N =10^7$  & 83.66 & 27.78 & 0.19  & 900+$^\dag$ & 44.51 & 1.89  & 900+$^\dag$ \\\hline 
    \end{tabular}%
  \label{tabTimevaryingN}%
\end{table}

\begin{table}[htbp]
  \centering
  \caption{Computation time (in mins) for the full data and PED, uniform, and twinning subdata with sample sizes, $n = 0.1N$ and $0.05N$. The number of observations are fixed at $p=100$ and the number of features vary from $10$ to $500$. The $^\dag$ represents that the program was exited after the said time.}
    \begin{tabular}{|l|r|r|r|r|r|r|r|}\hline
       &   & \multicolumn{3}{c|}{$n = 0.01N$} & \multicolumn{3}{c|}{$n = 0.05N$} \\
\cmidrule{2-8} &   \multicolumn{1}{c|}{Full}  & \multicolumn{1}{c|}{PED} & \multicolumn{1}{c|}{Uni} &  \multicolumn{1}{c|}{Twin}  & \multicolumn{1}{c|}{PED} & \multicolumn{1}{c|}{Uni} &\multicolumn{1}{c|}{Twin}  \\\hline

    10    & 0.94  & 0.77  & 0.01  & 0.51  & 0.81  & 0.02  & 1.80 \\
    100   & 10.99 & 3.33  & 0.02  & 33.81 & 4.69  & 0.13  & 166.52 \\
    500   & 60.73 & 25.95 & 0.10  & 110.62 & 30.38 & 0.70  & 459.66\\
   \hline 
    \end{tabular}%
  \label{tabTimevaryingK}%
\end{table}

\subsection{Real data}
We now consider two real datasets, one with a binary response and another with six classification categories.  

The binary research problem considered in the supersymmetric (SUSY) benchmark dataset \citep{baldi2014searching} distinguishes between a signal and a background process, where the signal process produces the supersymmetric particles. This dataset \citep{misc_susy_279} is available from the UCI Machine Learning Repository at the link: \url{https://archive.ics.uci.edu/ml/datasets/SUSY}. The dataset has 5M observations and 18 features. About 54\% of the observations correspond to the background processes. As suggested in the repository itself and \cite{wang2018optimal}, we consider the first 4.5M observations to be the full dataset for training, and the remaining 500,000 observations are considered as a test dataset. Table~\ref{tab-SUSY} reports the time (in mins) and the accuracy for the full dataset, PED, uniform, and twinning subdata. We consider sample sizes $1\%$ and $5\%$ of 4.5M for all subdata selection methods. First, note that the PED subdata has the best accuracy value among other subdata selection methods. The performance of the PED subdata is also closest to that of the full data. While one may claim that the differences are small, it is worth noticing that on a test dataset of size 500,000, a difference of $0.20\%$ between PED and Uni subdata implies the correct classification of $\sim$1018 extra processes, which we believe represents a significant improvement. 

\begin{table}[htbp]
  \centering
  \caption{Accuracy and time (in minutes) for the full dataset as well as the PED, uniform, twinning and IBOSS subdata with sample sizes $n = 0.1N$ and $0.05N$ on the SUSY dataset.}
    \begin{tabular}{|l|r|r|r|r|r|r|r|r|r|}\hline
       &   & \multicolumn{4}{c|}{$n = 0.01N$} & \multicolumn{4}{c|}{$n = 0.05N$} \\
\cmidrule{2-10} &   \multicolumn{1}{c|}{Full}  & \multicolumn{1}{c|}{PED} & \multicolumn{1}{c|}{Uni} &  \multicolumn{1}{c|}{Twin}  &\multicolumn{1}{c|}{IBOSS} & \multicolumn{1}{c|}{PED} & \multicolumn{1}{c|}{Uni} &\multicolumn{1}{c|}{Twin} &\multicolumn{1}{c|}{IBOSS}  \\\hline
Accuracy & 80.03 & 79.74 & 79.54 & 79.58 & 79.55 & 79.92 & 79.83 & 79.78 & 79.82  \\
    Time & 15.09 & 2.02  & 0.19  & 5.14  & 0.20  & 2.57  & 0.49  & 10.33 & 0.54    \\
    \hline \end{tabular}%
  \label{tab-SUSY} %
\end{table}

Next, we consider the forest cover type dataset \cite{misc_covertype_31}  for predicting the forest cover type from available cartographic variables. The original dataset contains 581012 observations, 10 continuous features, 41 binary features, and a response variable with seven categories. The seven categories with their respective proportions are 36.46\% (Spruce/Fir), 48.76\% (Lodgepole Pine), 6.15\% (Ponderosa Pine), 0.43\% (Cottonwood/Willow), 1.63\% (Aspen), 2.99\% (Douglas-fir) and 3.53\% (Krummholz). For IBOSS, we use the 10 continuous features as covariates, which measure geographical locations and lighting conditions. All 51 features are used for other methods. We use the 70\% of the dataset as training and the remaining 30\% as the test. We compute the accuracy and time for the full dataset, PED, uniform, and twinning subdata for twenty such training-test splits. The average accuracy and computation times across 20 repetitions are presented in Table~\ref{tab-Forest}. Results show that if the analysis can be done on the full data, then the subdata selection methods are not beneficial, neither from the computational perspective nor from the statistical. The twinning subdata has slightly higher accuracy than the twinning subdata for $5\%$ sample size, but the time taken by twinning is significantly higher than that taken by PED. Overall, the PED remains the superior choice among other subdata selection methods. 

\begin{table}[htbp]
  \centering
  \caption{Average Accuracy and time (in minutes) for the full dataset as well as the PED, uniform, twinning subdata with sample sizes $n = 0.1N$ and $0.05N$ on the forest cover type dataset.}
    \begin{tabular}{|l|r|r|r|r|r|r|r|r|r|}\hline
       &   & \multicolumn{3}{c|}{$n = 0.01N$} & \multicolumn{3}{c|}{$n = 0.05N$} \\
\cmidrule{2-8} &   \multicolumn{1}{c|}{Full}  & \multicolumn{1}{c|}{PED} & \multicolumn{1}{c|}{Uni} &  \multicolumn{1}{c|}{Twin}  & \multicolumn{1}{c|}{PED} & \multicolumn{1}{c|}{Uni} &\multicolumn{1}{c|}{Twin}  \\\hline
Accuracy & 96.32 & 78.89 & 77.24 & 78.01 & 86.07 & 85.06 & 86.09 \\
    Time &    0.43  & 0.49  & 0.05  & 2.11  & 0.56  & 0.06  & 10.15   \\
    \hline \end{tabular}%
  \label{tab-Forest} %
\end{table}

\section{Concluding remarks}
\label{sec:conc_remarks}
It is computationally challenging to build good prediction models for datasets with a large number of observations and a moderate number of features. For example, it takes about 83 minutes to build a random forest on the dataset with $10^7$ observations and 100 features. The time of 83 minutes is possible because of the complete parallelization of the random forest package in R; otherwise, the actual time would be much longer. The computers additionally need good memory storage for handling large datasets. Subdata selection is one emerging field that handles large datasets by selecting a representative big data sample. 

We are interested in a classification problem with an arbitrary number of categories. The existing subdata selection methods for classification are model-dependent since they are typically designed to achieve the best parameter estimates of the underlying generalized linear model. In practice, however, using advanced tools such as random forests for classification is more common. Using random forests as an analysis tool, we propose a novel model-free subdata selection method in this work. Our method, PED, first identifies a good partition of the big data by employing random small subsets of the big data. PED then selects an optimal sample from each stratum of the partition, where optimality is measured in terms of minimizing the expected Gini error on the test dataset. The optimal PED subdata is demonstrated to perform superior to other subdata selection methods with merely 1\% of 5\% of the big data size through extensive simulated and real datasets. We also observe that the performance of the PED subdata is close to that obtained by running the random forests on the full training data. The PED subdata is also computationally faster than other competing alternatives for subdata selection and for analyzing the big data, especially if either $N$ or $p$ is large.

\bigskip
\begin{center}
{\large\bf SUPPLEMENTARY MATERIAL}
\end{center}%

\begin{description}
\item The online Supplementary Material contains more performance results for the examples discussed in the main paper. The R code is available on request from the authors. 
\end{description}

\bibliographystyle{JASA}

\bibliography{Refs}
\end{document}